\documentclass[runningheads,a4paper]{llncs}
\usepackage{amssymb}
\usepackage{graphicx}
\usepackage{listings}
\usepackage{color}
\usepackage{float}
\usepackage[usenames,dvipsnames]{xcolor}
\newcommand{\ac}[1]{\textcolor{gray}{#1}}
\newcommand{\dc}[1]{\textcolor{red}{#1}}
\newcommand{\ec}[1]{\textcolor{blue}{#1}}

\newcommand{\mc}[1]{\textcolor{blue}{#1}}
\newcommand{\uc}[1]{\textcolor{red}{#1}}
\newcommand{\ic}[1]{\textcolor{green}{#1}}
\newcommand{\fc}[1]{\textcolor{black}{#1}}

\usepackage{url}
\urldef{\mailsa}\path|{francois.belanger3, aida.ouangraoua}@USherbrooke.ca|
  
\newcommand{\keywords}[1]{\par\addvspace\baselineskip
\noindent\keywordname\enspace\ignorespaces#1}

\begin{document}

\mainmatter  

\title{Alignment of protein-coding sequences with frameshift extension penalties}

\titlerunning{Alignment of coding sequences}

\author{Fran\c{c}ois B\'elanger
\and A\"ida Ouangraoua}

  
%
\authorrunning{B\'elanger-Ouangraoua}
\institute{Institute\\
\mailsa\\
}

%
%

\toctitle{Comparison of protein-coding DNA sequences}
\tocauthor{Bélanger-Ouangraoua}
\maketitle

\begin{abstract}
  We introduce an algorithm for the alignment of protein-coding sequences
  accounting for frameshifts. The main specificity of this algorithm as
  compared to previously published protein-coding sequence alignment
  methods is the introduction of a penalty cost for frameshift extensions.
  Previous algorithms have only used constant frameshift penalties.
  This is similar to the use of scoring schemes with affine gap penalties
  in classical sequence alignment algorithms. However, the overall penalty
  of a frameshift portion in an alignment cannot be formulated as an affine
  function, because it should also incorporate varying codon substitution
  scores.
  The second specificity of the algorithm is its search space being the set
  of all possible alignments between two coding sequences, under the
  classical definition of an alignment between two DNA sequences.
  Previous algorithms have introduced constraints on the length of the
  alignments, and additional symbols for the representation of frameshift
  openings in an alignment. The algorithm has the same asymptotic space and
  time complexity as the classical Needleman-Wunsch algorithm.
  

\end{abstract}

\section{Introduction and motivation}
\label{intro}

Comparative genomics is currently facing a huge challenge with the revelation
of a growing number of genes having multiple alternative coding sequences
in several species \cite{cunningham2015ensembl,pruitt2009consensus}. The various coding sequences arising from a same gene
or homologous genes differ not only by mutations in the nucleotide sequences,
but also by alternative start codons and alternative splicing of exons. All
these mechanisms often induce \emph{translation frameshifts} that lead to
different translations of a same portion of gene in distinct coding sequences
\cite{okamura2006frequent}.
This new enlightment on the complexity of gene architecture  evolution calls
for novel algorithms
for the comparison of coding sequences capable to account for the presence of
translation frameshifts between coding sequences.

The problem of aligning two coding sequences is an optimization problem that
consists in finding an optimal score alignment in a set of alignments
between the two sequences.
A \emph{coding sequence} is a DNA sequence composed of a succession of words of length $3$ called \emph{codons}.
An \emph{alignment} between two DNA sequences $A$ and $B$
is a pair of sequences $A'$ and $B'$ of same length $L$ on the alphabet of nucleotides
augmented with the gap symbol '-', such that $A'$ and $B'$ do not
contain a gap symbol '-' at a same
position, 
and $A$ and $B$ can be derived from  $A'$ and $B'$ by removing all the gap
symbols. The length $L$ of $A'$ and $B'$ is called the length of the
alignment.
A \emph{translation frameshift} in an alignment between two coding sequences is caused
by i) the deletion of one or two nucleotides of a codon (for example, a codon $\texttt{ACC}$ aligned with $\texttt{A--}$),
or ii) the insertion of nucleotides between two nucleotides of a codon
(for example, a codon $\texttt{A--CC}$ aligned with $\texttt{AGACC}$). 
The computation of an optimal alignment between two
coding sequences should account for both the translation of the coding
sequences into protein sequences, and the presence of translation
frameshifts between the two coding sequences.

A classical approach for comparing two coding sequences consists in a
three-step method, where coding sequences are first translated into protein
sequences, next protein sequences are aligned, and finally the protein
alignment is back-translated to a coding sequence alignment. This approach
is used in most tools for multiple alignment of coding sequences
\cite{abascal2010,wernersson2003,bininda2005,morgenstern2004}.
However, it is not able to account for the presence of frameshifts between
coding sequences.

The problem of aligning two coding sequences of length $n$ and $m$ while
accounting for both the corresponding protein sequences and the presence
of frameshifts was first addressed by Hein et al. \cite{hein1994,pedersen1998}.
They proposed a DNA/protein model such that the score of an alignment between
two coding sequences is a combination of its
score at the DNA level and its score at the protein level. Under this model,
a $O(n^2.m^2)$ algorithm \cite{hein1994} and then a $O(n.m)$
algorithm \cite{pedersen1998} were proposed
to compute an optimal score alignment. 
The search space of the algorithms are the set of alignments that can be each
uniquely decomposed into a succession of sub-alignments of eleven (11) types.
The eleven types of sub-alignment are defined such that the length of each
of them is a multiple of $3$. Thus, the total length of any alignment in the
search space is always a multiple of $3$, and the score of an alignment is  the sum of the scores of its sub-alignments.

Arvestad \cite{arvestad1997} proposed another $O(n.m)$ protein-coding
alignment algorithm
based on the concept of generalized substitutions introduced in
\cite{sankoff1983}. In this algorithm,  an alignment between two coding
sequences $A$ and $B$ is a pair of sequences on the alphabet of nucleotides
 augmented with the gap symbol '-' and the frameshift symbol '!'.
The search space of the algorithm is the set of alignments that are each
composed of a succession of sub-alignments of length $3$ such that each
sub-alignment is an alignment between two codon fragments of $A$ and $B$.
A \emph{codon fragment} of a coding sequence $S$ is defined as a
word of length $0$ to $5$ in $S$. If a codon fragment has a length of $4$
(resp. $5$), then one or two nucleotides in the codon fragment are
dropped in order to fit in a sub-alignment of length $3$. Such dropped
nucleotides are simply ignored in the definition of the score of a length-3
sub-alignment. If a codon fragment has a length of $1$ or $2$,
then two or one frameshift opening symbols '!' are added in the codon
in order to fit in a sub-alignment of length $3$.
The score of an alignment is then defined as the sum of the scores of its
length-3 sub-alignments.

More recently, Ranwez et al. \cite{ranwez2011} proposed  a simplification of
the model of Arvestad \cite{arvestad1997} where  a \emph{codon fragment}
of a coding sequence $S$ is defined as a word of length $0$ to $3$ in $S$.
Thus, no supplemental combinatorics are required in order to consider all the
possibilities of dropping one or two nucleotides from a codon fragment of
length $4$ or $5$. 
The algorithm has a complexity in $O(n.m)$. This method was extended in the
context of multiple protein-coding sequence alignment \cite{ranwez2011}.

The above three methods \cite{arvestad1997,pedersen1998,ranwez2011}
compare two coding sequences  while accounting
for the presence of translation frameshift openings between the two 
sequences. A frameshift in an alignment is penalized by adding
a constant frameshift cost, which only penalizes the initiation
of a frameshift, not accounting for the extension of this frameshift in
the alignment.

For example, we consider the following three coding sequences : Seq1, Seq2,
and Seq3. Seq1 has a length of $45$. Seq2 (resp. Seq3) has a length of $60$
and is obtained from Seq1 by deleting the nucleotide 'C' at position $30$
(nucleotide 'G' at position $15$) and adding $16$ nucleotides at the end.

\noindent
$\texttt{Seq1: \ac{ATG}ACC\ac{GAA}TCC\ac{AAG}CAG\ac{CCC}TGG\ac{CAT}AAG\ac{TGG}GGG\ac{AAC}GAT\ac{TGA}}$\\
$\texttt{~~~~~~~M~~T~~E~~S~~K~~Q~~P~~W~~H~~K~~W~~G~~N~~D~~*}$\\ 
$\texttt{Seq2: \ac{ATG}ACC\ac{GAA}TCC\ac{AAG}CAG\ac{CCC}TGG\ac{CAT}AAT\ac{GGG}GGA\ac{ACG}ATT\ac{GAA}GTA\ac{GGA}ACG\ac{ATT}TAA}$\\
$\texttt{~~~~~~~M~~T~~E~~S~~K~~Q~~P~~W~~H~~N~~G~~G~~T~~I~~E~~V~~G~~T~~I~~*}$\\ 
$\texttt{Seq3: \ac{ATG}ACC\ac{GAA}TCC\ac{AAC}AGC\ac{CCT}GGC\ac{ATA}AGT\ac{GGG}GGA\ac{ACG}ATT\ac{GAA}GTA\ac{GGA}ACG\ac{ATT}TAA}$\\
$\texttt{~~~~~~~M~~T~~E~~S~~N~~S~~P~~G~~I~~S~~G~~G~~T~~I~~E~~V~~G~~T~~I~~*}$

When looking at the translations of Seq1 and Seq2, it is easily observable
that Seq2 is more similar to Seq1, than Seq3 is similar to Seq1.
However, the pairwise alignment algorithms
accounting for frameshifts \cite{arvestad1997,pedersen1998,ranwez2011}
would return the same score for the two following optimal alignments
of Seq1 and Seq2, and Seq1 and Seq3, penalizing only the initiation
of a frameshift in both cases (positions colored in red in the alignments).

\noindent
$\texttt{\bf Optimal alignment between Seq1 and Seq2:}$\\
$\texttt{~M~~T~~E~~S~~K~~Q~~P~~W~~H~~\dc{K}~~\ec{W~~G~~N~~D~~*}~~-~~-~~-~~-~~-~~-}$\\
$\texttt{\ac{ATG}ACC\ac{GAA}TCC\ac{AAG}CAG\ac{CCC}TGG\ac{CAT}AAG\ac{TGG}GGG\ac{AAC}GAT\ac{TGA}------------------}$\\
$\texttt{\ac{ATG}ACC\ac{GAA}TCC\ac{AAG}CAG\ac{CCC}TGG\ac{CAT}AA-T\ac{GGG}GGA\ac{ACG}ATT\ac{GAA}GTA\ac{GGA}ACG\ac{ATT}TAA--}$\\
$\texttt{~M~~T~~E~~S~~K~~Q~~P~~W~~H~~\dc{!}~~\ec{W~~G~~N~~D~~*}~~S~~R~~N~~D~~L~~!~~}$\\
\noindent
$\texttt{\bf Optimal alignment between Seq1 and Seq3:}$\\
$\texttt{~M~~T~~E~~S~~\dc{K}~~\ec{Q~~P~~W~~H~~K~~W~~G~~N~~D~~*}~~-~~-~~-~~-~~-~~-}$\\
$\texttt{\ac{ATG}ACC\ac{GAA}TCC\ac{AAG}CAG\ac{CCC}TGG\ac{CAT}AAG\ac{TGG}GGG\ac{AAC}GAT\ac{TGA}------------------}$\\
$\texttt{\ac{ATG}ACC\ac{GAA}TCC\ac{AA-C}AGC\ac{CCT}GGC\ac{ATA}AGT\ac{GGG}GGA\ac{ACG}ATT\ac{GAA}GTA\ac{GGA}ACG\ac{ATT}TAA--}$\\
$\texttt{~M~~T~~E~~S~~\dc{!}~~\ec{Q~~P~~W~~H~~K~~W~~G~~N~~D~~*}~~S~~R~~N~~D~~L~~!}$

We describe a pairwise alignment algorithm that uses a scoring scheme
penalizing both the initiation and the extensions of frameshifts (positions
colored in blue  in the alignments).
In Section \ref{preliminaries}, some preliminary definitions of alignments and
the description of the problem are presented. In Section \ref{algorithm}, the
new algorithm for computing an optimal score alignment is described.

\newpage
\section{Preliminaries : Alignment of protein-coding sequences}
\label{preliminaries}

In this section, we formally describe coding sequences and the pairwise
alignment problem that is solved in Section \ref{algorithm}.

\begin{definition}[Coding sequence]
  A coding sequence is DNA sequence on the alphabet of nucleotides
  $\Sigma_N=\{a,c,g,t\}$ whose length $n$ is a multiple of $3$. A coding
  sequence is composed of a succession of $\frac{n}{3}$ codons that are
  the words of length $3$ in the sequence ending at positions $3i$,
  $1 \leq i \leq \frac{n}{3}$. The translation of the coding
  sequence is a protein sequence of length $\frac{n}{3}$ on the alphabet
  $\Sigma_A$ of amino acids (aa) such that each codon of the coding
  sequence is translated into an amino acid in the  protein sequence.
\end{definition}

In this work, the definition of an alignment between two coding sequences is
exactly the same as the classical definition of an alignment between two
DNA sequences used by the Needleman-Wunsch algorithm for the comparison of
two sequences \cite{needle}.

\begin{definition}[alignment between DNA sequences]
An alignment between two DNA sequences $A$ and $B$ is a pair
$(A',B')$ where $A'$ and $B'$ are two sequences of same length $L$
derived by inserting gap symbols $'-'$ in $A$ and $B$, such that
$\forall i, ~1 \leq i \leq L,  ~ A'[i] \neq '-'$ or $B'[i] \neq '-$.
Each position $i, ~1 \leq i \leq L$, in the alignment is called a
column of the alignment.
\end{definition}

Given a sequence $S$ of length $L$ on the alphabet  $\Sigma=\{a,c,g,t,-\}$,
$S[k~..~l], ~1 \leq k \leq l \leq L$, denotes the subsequence of $S$
going from position $k$ to position $l$. $|S[k~..~l]|$ denotes the number
of letters in $S[k~..~l]$ that are different from the gap symbol $'-'$.
For example, $|\texttt{AC--G}| = 3$.

Given an alignment $(A',B')$ between two coding sequences $A$ and $B$,
a codon of $A$ or $B$ is \emph{grouped in the alignment} if its
three nucleotides appear in three consecutive columns of the alignment.
For example, a codon $\texttt{ACC}$ that appears in the alignment as
$\texttt{ACC}$ is grouped, while it is not grouped if it appears as
$\texttt{A-CC}$.


In the following, we give our definition of the score of an alignment between
two coding sequences $A$ and $B$. It is based on a partition of the codons
of $A$ and $B$ into four sets (types):

The set of $\texttt{Matching codons (M)}$ contains the codons
that are grouped in the alignment, and aligned exactly with a
codon of the other sequence.

The set of $\texttt{Unmatching codons (U)}$ contains the codons that are
grouped in the alignment, and aligned with three consecutive nucleotides
of the other sequence that do not form a codon.

The set of $\texttt{Deleted/Inserted codons (InDel)}$ contains the codons
that are grouped in the alignment, and aligned with a succession of $3$ gaps.

All other codons are \emph{frameshift codons}. Following the definitions
and notations for frameshifts used in \cite{ranwez2011}, the set of frameshift
codons can be divided into two sets. The set of $\texttt{frameshift
  codons caused by deletions (FS$^-$)}$ contains the codons that are grouped
in the alignment, and are aligned with only one or two nucleotides in the
other sequence and some gap symbols.
The set of $\texttt{frameshift codons caused}$
  $\texttt{by insertions (FS$^+$)}$ contains
all the codons that are not grouped in the alignment.

The set of $\texttt{Matching nucleotides in frameshift codons (MFS)}$ contains
all the nucleotides belonging to a frameshift codon, and aligned with a
nucleotide of the other sequence.

The substitutions of matching (M)  and unmatching (U) codons  are scored
using an amino acid scoring function $s_{aa}$, and a fixed frameshift
extension cost denoted by  $\texttt{fs\_extension\_cost}$ is added for
each unmatching codon (U). The insertions/deletions of codons (Indel)
are scored by adding a
fixed gap cost denoted by $\texttt{gap\_cost}$ for each inserted/deleted
codon (Indel).
The alignment of frameshift codon nucleotides (MFS) are scored independently
from each other, using a nucleotide scoring function $s_{an}$. The insertions
or deletions of nucleotides from frameshift codons are responsible for
the initiation of frameshifts. They are then scored by adding a fixed
frameshift opening cost denoted by $\texttt{fs\_open\_cost}$ for each
frameshift codon.

In the following definition of the score of an alignment, the matching (M),
unmatching (U), and deleted/inserted
(InDel) codons of $A$ and $B$ are simply identified by the position (column)
of their
last nucleotide in the alignment. The matching nucleotides in frameshift codons
(MFS) are also identified by their positions in the alignment.

\begin{definition}[Score of an alignment]
  Let $(A',B')$ be an alignment of length $L$ between two coding sequences $A$ and $B$.
  
  \scriptsize
   \[ \begin{array}{lll}
  M_{A\rightarrow B} & = & \{ k, k \leq L ~|~ \exists ~ (i,j)  ~ s.t. ~ A'[k-2~..~k]=A[3i-2~..~3i] ~ and ~ B'[k-2~..~k]=B[3j-2~..~3j]\}\\
  &&\\
  U_{A\rightarrow B} &  = & \{ k, k \leq L ~|~  k \notin M_{A\rightarrow B} ~ and ~ \exists ~ i  ~ s.t. ~ A'[k-2~..~k]=A[3i-2~..~3i]  ~ and ~ |B'[k-2~..~k]| = 3\}\\
  &&\\
  Indel_{A\rightarrow B} & = & \{ k, k \leq L ~|~  \exists ~ i  ~ s.t. ~ A'[k-2~..~k]=A[3i-2~..~3i]~ and ~ |B'[k-2~..~k]| = 0\}\\
  &&\\
%
%
  MFS_{A\rightarrow B} & = & \{ k, k \leq L ~|~  \{k,k+1,k+2\} \cap  (M_{A\rightarrow B}\cup U_{A\rightarrow B}\cup InDel_{A\rightarrow B}) = \emptyset ~ and \\
  & & \exists ~ (i,j)  ~ s.t. ~ A'[k]=A[i] ~ and ~ B'[k]=B[j]\}\\
&& \\
  M_{B\rightarrow A} & = & \{ k, k \leq L ~|~ \exists ~ (j,i)  ~ s.t. ~ B'[k-2~..~k]=B[3j-2~..~3j] ~ and ~ A'[k-2~..~k]=A[3i-2~..~3i]\}\\
   &&\\ 
  U_{B\rightarrow A} &  = & \{ k, k \leq L ~|~  k \notin M_{B\rightarrow A} ~ and ~ \exists ~ j  ~ s.t. ~ B'[k-2~..~k]=A[3j-2~..~3j]  ~ and ~ |A'[k-2~..~k]| = 3\}\\
  &&\\
  Indel_{B\rightarrow A} & = & \{ k, k \leq L ~|~  \exists ~ j  ~ s.t. ~ B'[k-2~..~k]=B[3j-2~..~3j]~ and ~ |A'[k-2~..~k]| = 0\}\\
  &&\\
%
%
  MFS_{B\rightarrow A} & = & \{ k, k \leq L ~|~  \{k,k+1,k+2\} \cap  (M_{B\rightarrow A}\cup U_{B\rightarrow A}\cup InDel_{B\rightarrow A}) = \emptyset ~ and \\
  & & \exists ~ (j,i)  ~ s.t. ~ B'[k]=B[j] ~ and ~ A'[k]=A[i]\}
    \end{array} \]

  \normalsize
The score of the alignment $(A',B')$ is defined by :

    \[ \begin{array}{lll}
      \texttt{score}(A') & =& \sum_{k \in \texttt{M}_{A\rightarrow B}}{\frac{s_{aa}(A'[k-2~..~k],B'[k-2~..~k])}{2}} ~ +\\

      & & \sum_{k \in \texttt{U}_{A\rightarrow B}}{( \frac{s_{aa}(A'[k-2~..~k],B'[k-2~..~k])}{2} + \texttt{fs\_extension\_cost} )} ~ +\\

      & &  |\texttt{Indel}_{A\rightarrow B}| * \texttt{gap\_cost} ~ +\\

      & & (\frac{|A|}{3} -  |\texttt{M}_{A\rightarrow B}| - |\texttt{U}_{A\rightarrow B}| - |\texttt{InDel}_{A\rightarrow B}|) * \texttt{fs\_open\_cost} ~ + \\

      & & \sum_{k \in \texttt{MFS}_{A\rightarrow B}}{\frac{s_{an}(A'[k],B'[k])}{2}}
          \end{array} \]
\[ \begin{array}{lll}
  \texttt{score}(B') & =& \sum_{k \in \texttt{M}_{B\rightarrow A}}{\frac{s_{aa}(B'[k-2~..~k],A'[k-2~..~k])}{2}} +\\
   & & \sum_{k \in \texttt{U}_{B\rightarrow A}}{ (\frac{s_{aa}(B'[k-2~..~k],A'[k-2~..~k])}{2} + \texttt{fs\_extension\_cost})} +\\
  & &  |\texttt{Indel}_{B\rightarrow A}| * \texttt{gap\_cost} +\\
  & & (\frac{|B|}{3} -  |\texttt{M}_{B\rightarrow A}| - |\texttt{U}_{B\rightarrow A}| - |\texttt{InDel}_{B\rightarrow A}|) * \texttt{fs\_open\_cost} + \\  
  & & \sum_{k \in \texttt{MFS}_{B\rightarrow A}}{\frac{s_{an}(B'[k],A'[k])}{2}}\\
&& \\
  
    \texttt{score}(A',B') & = &  \texttt{score}(A') +  \texttt{score}(B')
    
    \end{array}
\]

    
\end{definition}

 




For example, consider  the two following sequences, $A$ containing
$13$ codons and $B$ containing $14$ codons,
and an alignment of length $48$ between them.

\noindent
$\texttt{A: \ac{ATG}ACC\ac{GAA}TCC\ac{AAG}CAG\ac{CCC}TGG\ac{CCA}GAT\ac{CAA}CGT\ac{TGA}}$\\
$\texttt{~~~~M~~T~~E~~S~~K~~Q~~P~~W~~P~~D~~Q~~R~~*}$\\
$\texttt{B: \ac{ATG}GAG\ac{TCG}AAG\ac{ATC}AGC\ac{TGG}CAG\ac{GCC}ATT\ac{GGC}AAT\ac{GAC}TGA}$\\
$\texttt{~~~~M~~E~~S~~K~~I~~S~~W~~Q~~A~~I~~G~~N~~D~~*~}$\\
\noindent
$\texttt{\bf An alignment (A',B') of length 48 between A and B:}$\\
$\texttt{pos 000000000111111111122222222223333333333444444444}$\\
$\texttt{~~~~123456789012345678901234567890123456789012345678}$\\
$\texttt{~~~~~M~~T~~E~~S~~K~~~~Q~~P~~W~~P~~~~~D~~~~~Q~~R~~~*}$\\
$\texttt{A'~~\mc{ATG}\ic{ACC}\mc{GAA}\mc{TCC}\mc{AAG}--\uc{CAG}\fc{CCC}\mc{TGG}\fc{CCA}\fc{G}---\fc{AT}---\uc{CAA}\fc{CG-T}\mc{TGA}}$\\
$\texttt{B'~~\mc{ATG}---\mc{GAG}\mc{TCG}\mc{AAG}\fc{ATC}\uc{AGC}--\mc{TGG}-\uc{CAG}\ic{GCC}\fc{ATT}\fc{GGC}\uc{AAT}\fc{GAC}\mc{TGA}}$\\
$\texttt{~~~~~M~~~~~E~~S~~K~~I~~S~~~~W~~~Q~~A~~I~~G~~N~~D~~*}$\\

The composition of the different sets of codons and nucleotides used in
the definition of the score of the alignment $(A',B')$ are:
\noindent
 \mc{$ \texttt{M}_{A\rightarrow B}  =  \{ 3, 9, 12, 15, 26, 48\} $};
 \uc{$ \texttt{U}_{A\rightarrow B}   =  \{20, 41\} $};
 \ic{$ \texttt{Indel}_{A\rightarrow B}  =  \{ 6 \} $}; 
$ \texttt{MFS}_{A\rightarrow B}  =  \{21,28,29,30,34,35,42,43,45\}$;
 \mc{$ \texttt{M}_{B\rightarrow A}  =  \{ 3, 9, 12, 15, 26, 48\}$};
 \uc{$ \texttt{U}_{B\rightarrow A}   =  \{ 21, 30, 42\}$};
 \ic{$ \texttt{Indel}_{B\rightarrow A}  =  \{ 33\}$}; and\\
 $ \texttt{MFS}_{B\rightarrow A}  =  \{ 18, 34, 35, 39, 43, 45\}. $


\section{Algorithm}
\label{algorithm}

In this section, we describe a $O(n.m)$ time and space complexity algorithm
that solves the problem of finding a maximum score alignment between two
coding sequences $A$ and $B$ of lengths $n$ and $m$.
Similarly to other sequence comparison methods \cite{needle,smith}, we use
dynamic programming tables of size $n+1\times m+1$ that are indexed by the
pairs of prefixes of the two coding sequences. The table $D$ stores the
maximum scores of the alignments between prefixes of $A$ and $B$. The table
$D_F$ is used to account for potential cases of frameshift extensions that
are counted subsequently.

\begin{definition}[Dynamic programming tables]
  \label{table}
  Given two coding sequences $A$ and $B$ as input, the algorithm uses
  two dynamic programming tables $D$ and $D_F$ of size $n+1\times m+1$.
  The cell $D(i,j)$ contains the maximum score of an alignment between
  the prefixes $A[1~..~i]$ and $B[1~..~j]$.
  The table $D_F$ is filled only for values of $i$ and $j$ such that
  $i (mod~3) = 0$ or $j (mod~3) = 0$.
  If $i (mod~3) \neq 0$
  (resp. $j (mod~3) \neq 0$), the cell $D_F(i,j)$ contains the score of an
  alignment between the prefixes $A[1~..~i+\alpha]$ and $B[1~..~j+\alpha]$ where
  $\alpha = (3-i) (mod~3)$ (resp. $\alpha = (3-j) (mod~3)$).
  The table $D_F$ is filled as follows:
  \begin{itemize}
    \item If $i (mod~3) = 0$ and $j (mod~3) = 0$, $D_F(i,j) = D(i,j)$.
    \item If $i (mod~3) = 0$ and $j (mod~3) = 2$, or $i (mod~3) = 2$
      and $j (mod~3) = 0$, $D_F(i,j)$ contains the maximum score of
      an alignment between $A[1~..~i+1]$ and $B[1~..~j+1]$ such that
      $A[i+1]$ and $B[j+1]$ are aligned together, and half of
      the score for aligning $A[i+1]$ with $B[i+1]$ is subtracted.
      
    \item If $i (mod~3) = 0$ and $j (mod~3) = 1$, or $i (mod~3) = 1$ and
      $j (mod~3) = 0$, $D_F(i,j)$ contains the maximum score of an alignment
      between $A[1~..~i+2]$ and $B[1~..~j+2]$ such that $A[i+1]$,$B[j+1]$ and
        $A[i+2]$,$B[j+2]$ are aligned together, and half
        of the scores of aligning $A[i+2]$, $B[i+2]$, and $A[i+1]$,
        $B[i+1]$ is subtracted. 
 \end{itemize}
\end{definition}

\begin{lemma}[Filling up table D]
  \label{D}
  \begin{enumerate}
  \item {\bf If $i (mod~3) = 0$ and $j (mod~3) = 0$}
    \scriptsize
      \[ D(i,j) = \max \left\{
  \begin{array}{ll}
    1. & s_{aa}(A[i-2~..~i],B[j-2~..~j]) + D(i-3,j-3)\\
    
    2. & s_{an}(A[i],B[j]) + s_{an}(A[i-1],B[j-1]) + D(i-3,j-2) + 2 * \texttt{fs\_open\_cost}\\
    
    3. & s_{an}(A[i],B[j]) + s_{an}(A[i-2],B[j-1]) + D(i-3,j-2) + 2 * \texttt{fs\_open\_cost}\\
    
    4. & s_{an}(A[i],B[j]) + D(i-3,j-1) + 2* \texttt{fs\_open\_cost}\\
    
    5. & s_{an}(A[i],B[j]) + s_{an}(A[i-1],B[j-1]) + D(i-2,j-3) + 2 * \texttt{fs\_open\_cost}\\
    
    6. & s_{an}(A[i],B[j]) + s_{an}(A[i-1],B[j-2]) + D(i-2,j-3) + 2 * \texttt{fs\_open\_cost}\\

    7. & s_{an}(A[i],B[j]) + D(i-1,j-3) + 2* \texttt{fs\_open\_cost}\\

    8. & s_{an}(A[i],B[j]) + D(i-1,j-1) + 2 * \texttt{fs\_open\_cost}\\

    9. & \frac{s_{an}(A[i-1],B[j])}{2} + \frac{s_{an}(A[i-2],B[j-1])}{2} + D_F(i-3,j-2) + \texttt{fs\_open\_cost}\\
    
    10. & s_{an}(A[i-1],B[j]) + D(i-3,j-1) + 2 * \texttt{fs\_open\_cost}\\
    
    11. & \frac{s_{an}(A[i-2],B[j])}{2} + D_F(i-3,j-1) + \texttt{fs\_open\_cost}\\
    
    12. & \texttt{gap\_cost} + D(i-3,j) \\
    
    13. & D(i-1,j) + \texttt{fs\_open\_cost}\\
    
    14. & \frac{s_{an}(A[i],B[j-1])}{2} + \frac{s_{an}(A[i-1],B[j-2])}{2} + D_F(i-2,j-3) + \texttt{fs\_open\_cost}\\
    
    15. & s_{an}(A[i],B[j-1]) + D(i-1,j-3) + 2 * \texttt{fs\_open\_cost}\\
    
    16. & \frac{s_{an}(A[i],B[j-2])}{2} + D_F(i-1,j-3) + \texttt{fs\_open\_cost}\\
    
    17. & \texttt{gap\_cost} + D(i,j-3) \\
    
    18. & D(i,j-1) + \texttt{fs\_open\_cost}\\  
  \end{array}
\right.
\]

    \normalsize
    \item {\bf If $i (mod~3) = 0$ and $j (mod~3) \neq 0$}

      \scriptsize
      \[ D(i,j) = \max \left\{
  \begin{array}{ll}
  1. &   \frac{s_{aa}(A[i-2~..~i],B[j-2~..~j])}{2} + D_F(i-3,j-3) + \texttt{fs\_extension\_cost}\\
   & + \frac{s_{an}(A[i],B[j])}{2} (+ \frac{s_{an}(A[i-1],B[j-1])}{2} ~if ~j-1 (mod~3) \neq 0)\\
    
  2. &   s_{an}(A[i],B[j]) + s_{an}(A[i-1],B[j-1]) + D(i-3,j-2) + \texttt{fs\_open\_cost} \\
   & (+ \texttt{fs\_open\_cost} ~if ~j-1 (mod~3) = 0)\\

  3. &  s_{an}(A[i],B[j]) + s_{an}(A[i-2],B[j-1]) + D_F(i-3,j-2) + \texttt{fs\_open\_cost}\\
  & (- \frac{s_{an}(A[i-2],B[j-1])}{2} ~if ~j-1 (mod~3) = 0)\\
    
  4. &  s_{an}(A[i],B[j]) + D(i-3,j-1) + \texttt{fs\_open\_cost}\\
      
  5. &  s_{an}(A[i],B[j]) + D(i-1,j-1) + \texttt{fs\_open\_cost}\\

  6. &  s_{an}(A[i-1],B[j]) + s_{an}(A[i-2],B[j-1]) + D_F(i-3,j-2) + \texttt{fs\_open\_cost}\\
   & (- \frac{s_{an}(A[i-2],B[j-1])}{2} ~if ~j-1 (mod~3) = 0)\\
      
  7. & s_{an}(A[i-1],B[j]) + D(i-3,j-1) + \texttt{fs\_open\_cost}\\
    
  8. &  s_{an}(A[i-2],B[j]) + D(i-3,j-1) + \texttt{fs\_open\_cost}\\

  9. &   \texttt{gap\_cost} + D(i-3,j) \\

  10. & D(i-1,j) + \texttt{fs\_open\_cost}\\

  11. &  D(i,j-1)\\
    
  \end{array}
\right.
\]

    \normalsize
      \item {\bf If $i (mod~3) \neq 0$ and $j (mod~3) = 0$}, the equation is symmetric to the previous case.

      \item {\bf If $i (mod~3) \neq 0$ and $j (mod~3) \neq 0$}

      \scriptsize
      $ D(i,j) = \max \left\{
  \begin{array}{ll}
    1. & s_{an}(A[i],B[j]) + D(i-1,j-1)\\

    2. & D(i-1,j)\\
      
    3. & D(i,j-1)\\    
  \end{array}
\right.
$
      \normalsize

  \end{enumerate}
\end{lemma}

\newpage

\begin{figure}[H]
\centering
Case 1. $i (mod~3) = 0$ and $j (mod~3) = 0$
\includegraphics[width=\textwidth]{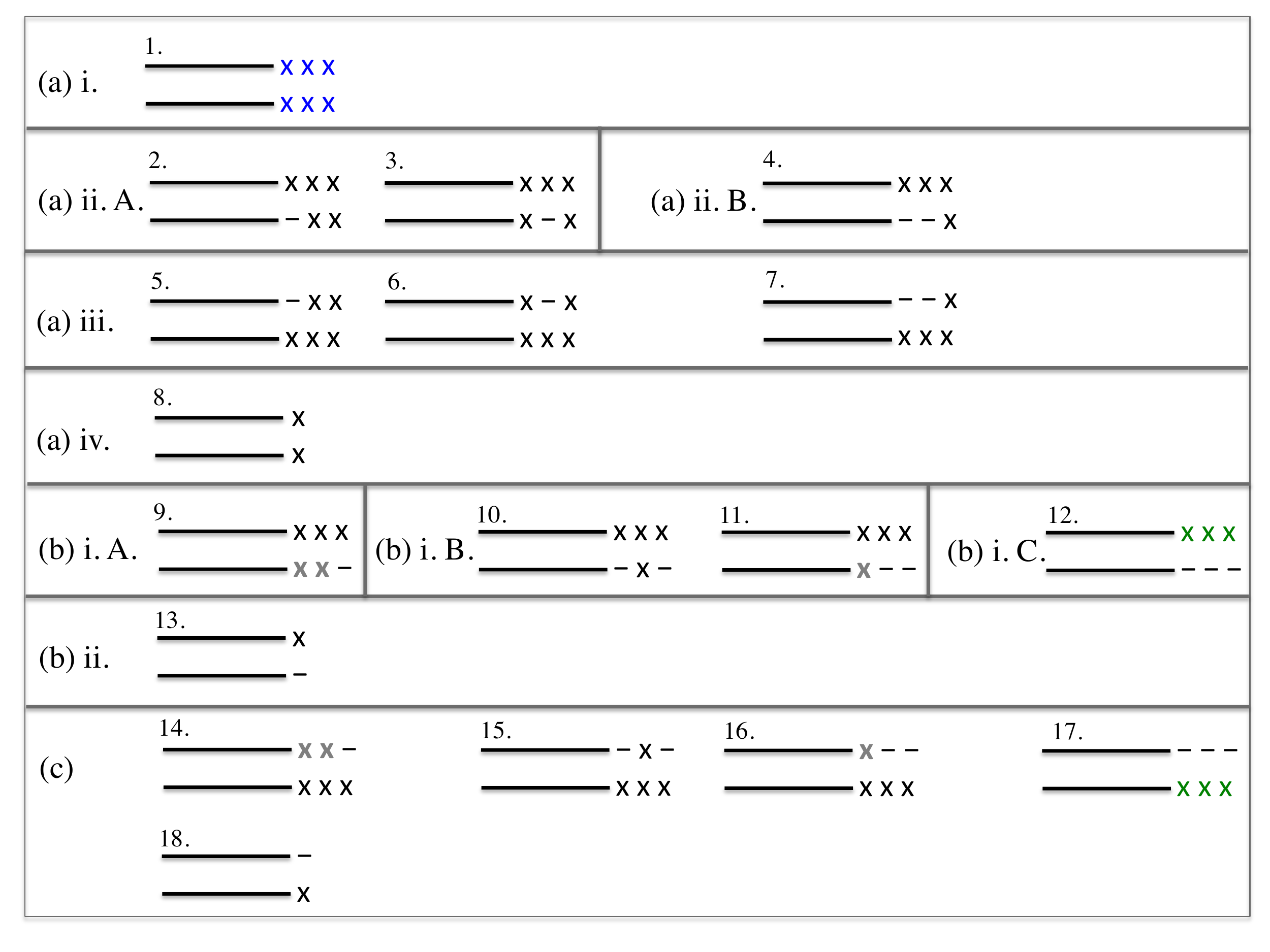}
Case 2. $i (mod~3) = 0$ and $j (mod~3) \neq 0$
\includegraphics[width=\textwidth]{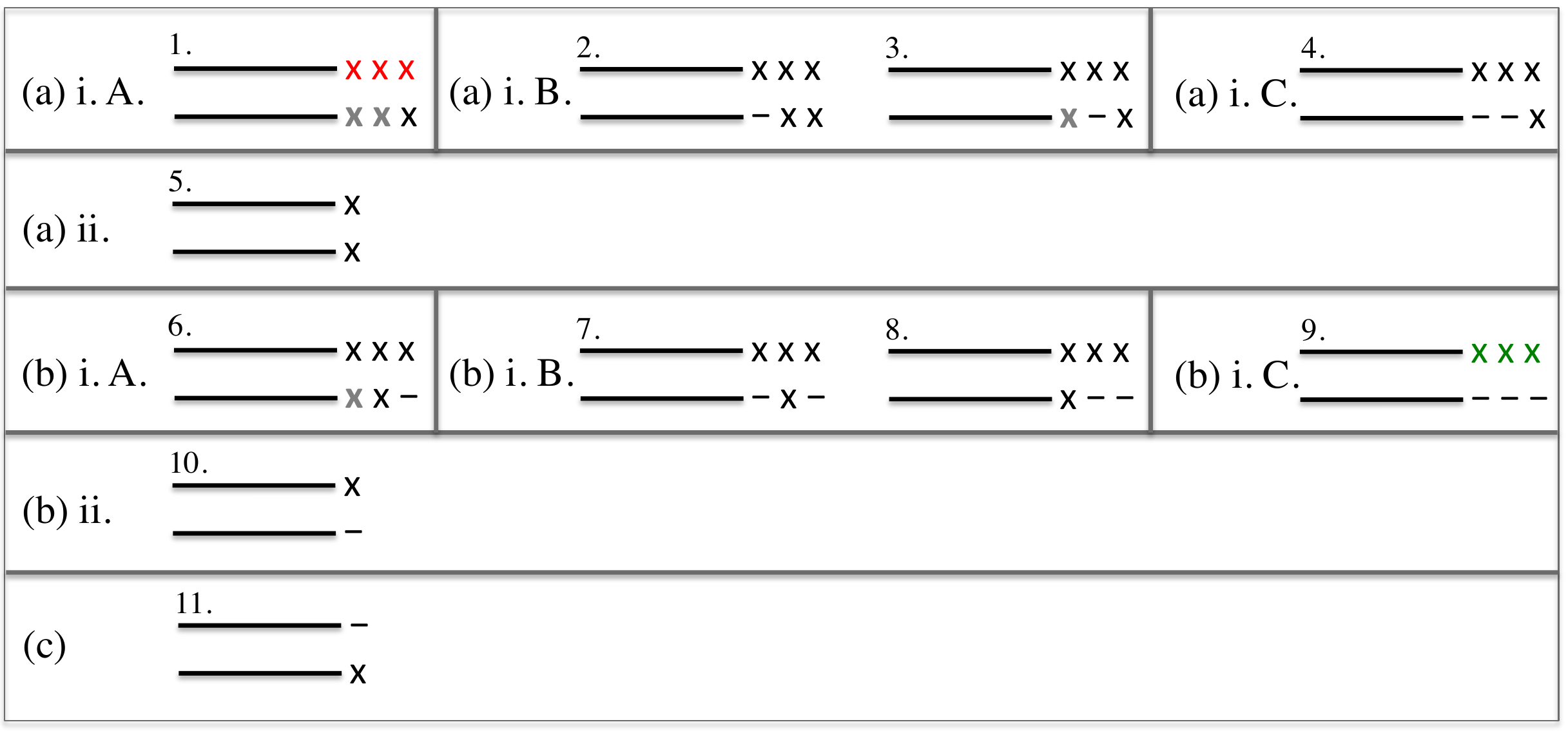}
\caption{Illustration of the configurations
  of alignment considered in Lemma \ref{D}
  for computing $D(i,j)$ in the cases 1 and 2.
  The right-most nucleotides of the sequences $A[1~..~i]$
  and $B[1~..~j]$ are represented using the character $\texttt{x}$.
  The nucleotides are colored according to the type of the codon
  to which they belong : matching codons (M) in blue color, unmatching
  codons (U) in red color, inserted/deleted codons (Indel) in green color,
  and frameshift codons (FS) in black color. The nucleotides that
  appear in gray color are those belonging to codons whose type has not
  yet been decided. In such case, the table $D_F$ is used in order
  to decide of the type of these codons later, and adjust the score accordingly.
}
\label{fig:D1}
\end{figure}

\newpage
\begin{proof}[Proof of Lemma \ref{D}] The principle of the proof is
  similar to the one for the alignment of non-coding sequences \cite{needle}.
  For each case, the score $D(i,j)$ is the maximum score of all possible
  alignment configurations that are considered for this case. Here, we only
  describe the alignment configurations considered in the case 1 where
  $i (mod~3) = 0$ and $j (mod~3) = 0$.   A complete proof for all the cases
  of the Lemma is given in Appendix. An illustration of the different configurations of
  alignment considered for the cases 1 and 2 is shown in Figure
  \ref{fig:D1}.

    \begin{enumerate}
   \item {\bf If $i (mod~3) = 0$ and $j (mod~3) = 0$}, there are
     three cases depending on the alignment of $A[i]$ and $B[j]$.
      \begin{enumerate}
      \item {\bf If $A[i]$ and $B[j]$ are aligned together}, there are four cases
        depending on whether $A[i-2~..~i]$ and $B[j-2~..~j]$ are grouped in
        the alignment or not.
        \begin{enumerate}
        \item {\bf If both $A[i-2~..~i]$ and $B[j-2~..~j]$ are grouped}, then
          $A[i-2~..~i]$ and $B[j-2~..~j]$  have to be aligned together, and the
          score of the alignment is:\\
          1.  $s_{aa}(A[i-2~..~i],B[j-2~..~j]) + D(i-3,j-3)$
        \item {\bf If  $A[i-2~..~i]$ is grouped while  $B[j-2~..~j]$ is not grouped},
          then both $A[i-2~..~i]$ and $B[j-2~..~j]$ are FS codons ($A[i-2~..~i]$ is
          a FS$^-$ codon while $B[j-2~..~j]$ is a FS$^+$ codon). We add
          $2 * \texttt{fs\_open\_cost}$ to the score of the alignment, and
          the alignment of the nucleotides of the two FS codons can be scored
          independently using the scoring function $s_{an}$.
         There are two cases depending on the number of  nucleotides from
          $B[j-2~..~j]$ that are aligned with $A[i-2~..~i]$, two or one:
            \begin{enumerate}
            \item {\bf If  $A[i-2~..~i]$ is aligned with two nucleotides}, then these
              nucleotides are $B[j-1]$ and $B[j]$.
              There are two cases depending on the alignment of the nucleotide
              $B[j-1]$ with $A[i-1]$ or $A[i-2]$:\\
              2. $s_{an}(A[i],B[j]) + s_{an}(A[i-1],B[j-1]) + D(i-3,j-2) + 2 * \texttt{fs\_open\_cost}$\\
              3. $s_{an}(A[i],B[j]) + s_{an}(A[i-2],B[j-1]) + D(i-3,j-2) + 2 * \texttt{fs\_open\_cost}$
            \item {\bf If  $A[i-2~..~i]$ is aligned with one nucleotide},  then
              this single nucleotide is $B[j]$, and the score of the alignment is:\\
             4. $s_{an}(A[i],B[j]) + D(i-3,j-1) + 2 * \texttt{fs\_open\_cost}$
            \end{enumerate}
          \item  {\bf If  $A[i-2~..~i]$ is not grouped while $B[j-2~..~j]$ is grouped},
            there are three cases that are symmetric to the three cases from
            (a)ii.:\\
            5. $s_{an}(A[i],B[j]) + s_{an}(A[i-1],B[j-1]) + D(i-2,j-3) + 2 * \texttt{fs\_open\_cost}$\\
            6. $s_{an}(A[i],B[j]) + s_{an}(A[i-1],B[j-2]) + D(i-2,j-3) + 2 * \texttt{fs\_open\_cost}$\\
            7. $s_{an}(A[i],B[j]) + D(i-1,j-3) + 2 * \texttt{fs\_open\_cost}$
          \item {\bf If both $A[i-2~..~i]$ and $B[j-2~..~j]$ are not grouped}, then again
        both $A[i-2~..~i]$ and $B[j-2~..~j]$ are FS codons (both are FS$^+$ codons):\\
        8. $s_{an}(A[i],B[j]) + D(i-1,j-1) + 2 * \texttt{fs\_open\_cost}$\\
        \end{enumerate}
        
      \item {\bf If $A[i]$ is aligned with a gap}, then the codon  $A[i-2~..~i]$
        is a FS codon (FS$^-$ or FS$^+$). We must add $\texttt{fs\_open\_cost}$
        to the score of
        the alignment. There are two cases depending on whether
        $A[i-2~..~i]$ is grouped in the alignment or not.
        \begin{enumerate}
        \item {\bf If $A[i-2~..~i]$ is grouped}, then there are three cases depending
          on the number of  nucleotides from $B[j-2~..~j]$ that are aligned
          with $A[i-2~..~i]$, two, one, or zero.
            \begin{enumerate}
            \item {\bf If  $A[i-2~..~i]$ is aligned with two nucleotides}, then
              these
              nucleotides are $B[j-1]$ and $B[j]$. The score of the alignment is:\\
              9. $\frac{s_{an}(A[i-1],B[j])}{2} + \frac{s_{an}(A[i-2],B[j-1])}{2} + D_F(i-3,j-2) + \texttt{fs\_open\_cost}$
            \item {\bf If  $A[i-2~..~i]$ is aligned with one nucleotide}, then this
              single nucleotide is $B[j]$. There two cases depending on
              the alignment of the nucleotide $B[j]$ with $A[i-1]$ or $A[i-2]$:\\
              10. $s_{an}(A[i-1],B[j]) + D(i-3,j-1) + 2 * \texttt{fs\_open\_cost}$\\
              11. $\frac{s_{an}(A[i-2],B[j])}{2} + D_F(i-3,j-1) + \texttt{fs\_open\_cost}$
            \item {\bf If $A[i-2~..~i]$ is aligned with zero nucleotide}, then
              the codon
        $A[i-2~..~i]$ is entirely deleted. The score of the alignment is: \\
              12. $\texttt{gap\_cost} + D(i-3,j)$ 
            \end{enumerate}
          \item  {\bf If $A[i-2~..~i]$ is not grouped}, then the codon $A[i-2~..~i]$
            is a FS$^+$ codon, and the score of the alignment is:\\
        13. $D(i-1,j) + \texttt{fs\_open\_cost}$\\  
        \end{enumerate}

      \item {\bf If $B[i]$ is aligned with a gap}, there are fives cases that
        are symmetric
        to the five cases from (b):\\
        14. $\frac{s_{an}(A[i],B[j-1])}{2} + \frac{s_{an}(A[i-1],B[j-2])}{2} + D_F(i-2,j-3) + \texttt{fs\_open\_cost}$\\
        15. $s_{an}(A[i],B[j-1]) + D(i-1,j-3) + 2 * \texttt{fs\_open\_cost}$\\  
        16. $\frac{s_{an}(A[i],B[j-2])}{2} + D_F(i-1,j-3) + \texttt{fs\_open\_cost}$\\
         17. $\texttt{gap\_cost} + D(i,j-3)$\\
         18. $D(i,j-1) + \texttt{fs\_open\_cost}$\\  
      \end{enumerate}

  \end{enumerate}

\end{proof}

\begin{lemma}[Filling up table $D_F$]
    \label{DF}
  \begin{enumerate}
  \item {\bf If $i (mod~3) = 0$ and $j (mod~3) = 0$}\\
    \scriptsize
    $ D_F(i,j) = D(i,j)$\\
    
    \normalsize
  \item {\bf If $i (mod~3) = 2$ and $j (mod~3) = 0$}\\
    \scriptsize    
      \[ D_F(i,j) = \max \left\{
  \begin{array}{ll}
    1. & \frac{s_{aa}(A[i-1~..~i+1],B[j-1~..~j+1])}{2} + D_F(i-2,j-2) + \texttt{fs\_extension\_cost}\\
    
    2. & \frac{s_{an}(A[i+1],B[j+1])}{2} + s_{an}(A[i],B[j]) + D(i-2,j-1) + 2 * \texttt{fs\_open\_cost}\\
    
    3. & \frac{s_{an}(A[i+1],B[j+1])}{2} + \frac{s_{an}(A[i-1],B[j])}{2} + D_F(i-2,j-1) + \texttt{fs\_open\_cost}\\

    4. & \frac{s_{an}(A[i+1],B[j+1])}{2}  + D(i-2,j) + \texttt{fs\_open\_cost}\\

    5. & \frac{s_{an}(A[i+1],B[j+1])}{2}  + D(i,j) + \texttt{fs\_open\_cost}\\
  \end{array}
\right.
\]
    \normalsize
  \item {\bf If $i (mod~3) = 0$ and $j (mod~3) = 2$}, the equation is symmetric to the previous case.\\

\item {\bf If $i (mod~3) = 1$ and $j (mod~3) = 0$}\\
    \scriptsize    
      \[ D_F(i,j) = \max \left\{
  \begin{array}{ll}
    1. & \frac{s_{aa}(A[i~..~i+2],B[j~..~j+2])}{2} + D_F(i-1,j-1) + \texttt{fs\_extension\_cost}\\
    
    2. & \frac{s_{an}(A[i+2],B[j+2])}{2} + \frac{s_{an}(A[i+1],B[j+1])}{2}  + D(i-1,j) + \texttt{fs\_open\_cost}\\
    
    3. & \frac{s_{an}(A[i+2],B[j+2])}{2} + \frac{s_{an}(A[i+1],B[j+1])}{2}  + D(i,j) + \texttt{fs\_open\_cost}\\
  \end{array}
\right.
\]
    \normalsize

  \item {\bf If $i (mod~3) = 0$ and $j (mod~3) = 1$}, the equation is symmetric to the previous case.
  \end{enumerate}
\end{lemma}

The proof of Lemma   \ref{DF} follows from Lemma \ref{D}. It is given in Appendix.
We  now present the alignment algorithm using Lemma \ref{D}  and  \ref{DF}
in the next  theorem.

\begin{theorem}
  \label{thm}
  Given two coding sequences $A$ and $B$ of lengths $n$ and $m$,
  a maximum score alignment between $A$ and $B$ can be found in
  time and space  $O(n \times m)$, using the following algorithm.\\

  \scriptsize    
  \noindent
  $\texttt{Algorithm Align(A,B)}$\\
  $\texttt{~~for i = 0 to n do}$\\
      $\texttt{~~~~~~}D(i,0) = floor(\frac{i}{3}) * \texttt{gap\_cost}$\\
      $\texttt{~~~~~~}D_F(i,0) = D(i,0) + \left\{
   \begin{array}{ll}
    \frac{s_{an}(A[i+1],B[1])}{2} + \frac{s_{an}(A[i+2],B[2])}{2} + \texttt{fs\_open\_cost}, &  \texttt{if ~i (mod~3) = 1}\\
    \frac{s_{an}(A[i+1],B[1])}{2}  + \texttt{fs\_open\_cost}, &  \texttt{if ~i (mod~3) = 2}
  \end{array}
   \right.$
   
  \noindent   
  $\texttt{~~for j = 0 to m do}$\\
      $\texttt{~~~~~~}D(0,j) = floor(\frac{j}{3}) * \texttt{gap\_cost}$\\
      $\texttt{~~~~~~}D_F(0,j) = D(0,j) + \left\{
   \begin{array}{ll}
    \frac{s_{an}(A[1],B[j+1])}{2} + \frac{s_{an}(A[2],B[j+2])}{2} + \texttt{fs\_open\_cost}, &  \texttt{if ~j (mod~3) = 1}\\
    \frac{s_{an}(A[1],B[j+1])}{2}  + \texttt{fs\_open\_cost}, &  \texttt{if ~j (mod~3) = 2}\\
  \end{array}
\right.$
  $\texttt{~~for i = 0 to n do}$\\
      $\texttt{~~~~~~for j = 0 to m do}$\\
          $\texttt{~~~~~~~~~~compute D(i,j) using Lemma \ref{D}}$\\
          $\texttt{~~~~~~~~~~compute $D_F$(i,j) using Lemma \ref{DF}, ~if ~i (mod~3) = 0 ~or ~j (mod~3) = 0}$
 
\end{theorem}

The proof of Theorem \ref{thm} is given in Appendix.

 \vspace{-0.2cm}
 \section{Implementation}

 We implemented the algorithm presented in this paper and the pairwise
 alignment algorithm accounting for frameshift opening penalties described
 in \cite{ranwez2011}.

 We applied both algorithms to the alignment of the examples of coding sequences
 Seq1, Seq2, and Seq3 described in Section \ref{intro}, with the following
 parameters: $\texttt{gap\_cost} = -1$, $\texttt{fs\_open\_cost} = -2$,
 $\texttt{fs\_extension\_cost} = -1$,  $s_{aa}$ corresponding to the amino acid
 substitution matrix BLOSUM62, and $s_{an}$ returning a score of $+1$
 (resp. $-1$) for a match (resp. mismatch) between two nucleotides.
 As predicted, the application of the algorithm from \cite{ranwez2011}
 to  Seq1, Seq2 and Seq3  yields the same score of $72.0$ for both the alignment
 between Seq1 and Seq2, and  the alignment between Seq1 and Seq3. The present
 algorithm yields a score of $68.5$ for Seq1 and Seq2,
 and a lower score of $58.0$ for Seq1 and Seq3.

 Using the same parameters, both algorithms were also applied to pairs of
 human coding sequences from paralogous genes that share a common coding
 subsequence translated in different frames (see \cite{okamura2006frequent}
 for a list
 of $470$ pairs of human coding sequences presenting a frameshift event).
 In Appendix, the alignments obtained for the coding
 sequences of the protein \emph{NM\_001083537} from Gene FAM86B1 and the protein
 \emph{NM\_018172} from Gene FAM86C1 are shown. These alignments show that both coding sequences
 share a common prefix subsequence translated in the same frame, and a
 common subsequence at the end of  \emph{NM\_018172}  translated in different frames,
 yielding a frameshift event. The algorithm of \cite{ranwez2011} yields
 a high score of $718.0$ for the alignment, while the present
 algorithm return a score of $530$ accounting for a frameshift extension
 length of 81 nucleotides.

 \vspace{-0.2cm}
\section{Conclusion}
\label{conclusion}

We introduce a new algorithm for the pairwise alignment protein-coding
sequences, accounting for translation frameshift extensions and their
consequences on the modification of the protein sequences. The dynamic
programming algorithm has the same asymptotic space and time complexity
as the classical Needleman-Wunsch algorithm.
The perspectives of this work include the evaluation of the impact of
the new method on the comparison of pairs of coding sequences listed in
biological databases. We also plan to study the extension of the method
in the context of multiple protein-coding sequence alignment.

\vspace{-0.2cm}
\bibliographystyle{plain}

\bibliography{biblio}

\newpage
\section*{Appendix}



\subsection*{Complete proof of Lemma \ref{D}}

\begin{proof}[Complete proof of Lemma \ref{D}]
  An illustration of  the different configurations of alignment
  considered for the cases 1 and 2 of Lemma \ref{D} in this proof
  is given in Figure \ref{fig:D1}.
  For each of the cases 1, 2, 3 and 4 of the Lemma,
  we first consider three cases depending on the configurations of the
  alignment of $A[i]$ and $B[j]$: (a) $A[i]$ and $B[j]$
  are aligned together, (b) $A[i]$ is aligned with a gap, (c) $B[j]$ is
  aligned with a gap.
    \begin{enumerate}
   \item {\bf If $i (mod~3) = 0$ and $j (mod~3) = 0$}, then $A[i]$ and $B[j]$ are
     the last nucleotides of two codons $A[i-2~..~i]$ and $B[j-2~..~j]$. There are
     three cases depending on the alignment of $A[i]$ and $B[j]$.
      \begin{enumerate}
      \item {\bf If $A[i]$ and $B[j]$ are aligned together}, there are four cases
        depending on whether $A[i-2~..~i]$ and $B[j-2~..~j]$ are grouped in
        the alignment or not.
        \begin{enumerate}
        \item {\bf If both $A[i-2~..~i]$ and $B[j-2~..~j]$ are grouped}, then
          $A[i-2~..~i]$ and $B[j-2~..~j]$  have to be aligned together, and the
          score of the alignment is:\\
          1.  $s_{aa}(A[i-2~..~i],B[j-2~..~j]) + D(i-3,j-3)$
        \item {\bf If  $A[i-2~..~i]$ is grouped while  $B[j-2~..~j]$ is not grouped},
          then both $A[i-2~..~i]$ and $B[j-2~..~j]$ are FS codons ($A[i-2~..~i]$ is
          a FS$^-$ codon while $B[j-2~..~j]$ is a FS$^+$ codon). We add
          $2 * \texttt{fs\_open\_cost}$ to the score of the alignment, and
          the alignment of the nucleotides of the two FS codons can be scored
          independently using the scoring function $s_{an}$.
         There are two cases depending on the number of  nucleotides from
          $B[j-2~..~j]$ that are aligned with $A[i-2~..~i]$, two or one:
            \begin{enumerate}
            \item {\bf If  $A[i-2~..~i]$ is aligned with two nucleotides}, then these
              nucleotides are $B[j-1]$ and $B[j]$.
              There are two cases depending on the alignment of the nucleotide
              $B[j-1]$ with $A[i-1]$ or $A[i-2]$:\\
              2. $s_{an}(A[i],B[j]) + s_{an}(A[i-1],B[j-1]) + D(i-3,j-2) + 2 * \texttt{fs\_open\_cost}$\\
              3. $s_{an}(A[i],B[j]) + s_{an}(A[i-2],B[j-1]) + D(i-3,j-2) + 2 * \texttt{fs\_open\_cost}$
            \item {\bf If  $A[i-2~..~i]$ is aligned with one nucleotide},  then
              this single nucleotide is $B[j]$, and the score of the alignment is:\\
             4. $s_{an}(A[i],B[j]) + D(i-3,j-1) + 2 * \texttt{fs\_open\_cost}$
            \end{enumerate}
          \item  {\bf If  $A[i-2~..~i]$ is not grouped while $B[j-2~..~j]$ is grouped},
            there are three cases that are symmetric to the three cases from
            (a)ii.:\\
            5. $s_{an}(A[i],B[j]) + s_{an}(A[i-1],B[j-1]) + D(i-2,j-3) + 2 * \texttt{fs\_open\_cost}$\\
            6. $s_{an}(A[i],B[j]) + s_{an}(A[i-1],B[j-2]) + D(i-2,j-3) + 2 * \texttt{fs\_open\_cost}$\\
            7. $s_{an}(A[i],B[j]) + D(i-1,j-3) + 2 * \texttt{fs\_open\_cost}$
          \item {\bf If both $A[i-2~..~i]$ and $B[j-2~..~j]$ are not grouped}, then again
        both $A[i-2~..~i]$ and $B[j-2~..~j]$ are FS codons (both are FS$^+$ codons):\\
        8. $s_{an}(A[i],B[j]) + D(i-1,j-1) + 2 * \texttt{fs\_open\_cost}$\\
        \end{enumerate}
        
      \item {\bf If $A[i]$ is aligned with a gap}, then the codon  $A[i-2~..~i]$
        is a FS codon (FS$^-$ or FS$^+$). We must add $\texttt{fs\_open\_cost}$
        to the score of
        the alignment. There are two cases depending on whether
        $A[i-2~..~i]$ is grouped in the alignment or not.
        \begin{enumerate}
        \item {\bf If $A[i-2~..~i]$ is grouped}, then there are three cases depending
          on the number of  nucleotides from $B[j-2~..~j]$ that are aligned
          with $A[i-2~..~i]$, two, one, or zero.
            \begin{enumerate}
            \item {\bf If  $A[i-2~..~i]$ is aligned with two nucleotides}, then
              these
              nucleotides are $B[j-1]$ and $B[j]$. The score of the alignment is:\\
              9. $\frac{s_{an}(A[i-1],B[j])}{2} + \frac{s_{an}(A[i-2],B[j-1])}{2} + D_F(i-3,j-2) + \texttt{fs\_open\_cost}$
            \item {\bf If  $A[i-2~..~i]$ is aligned with one nucleotide}, then this
              single nucleotide is $B[j]$. There two cases depending on
              the alignment of the nucleotide $B[j]$ with $A[i-1]$ or $A[i-2]$:\\
              10. $s_{an}(A[i-1],B[j]) + D(i-3,j-1) + 2 * \texttt{fs\_open\_cost}$\\
              11. $\frac{s_{an}(A[i-2],B[j])}{2} + D_F(i-3,j-1) + \texttt{fs\_open\_cost}$
            \item {\bf If $A[i-2~..~i]$ is aligned with zero nucleotide}, then
              the codon
        $A[i-2~..~i]$ is entirely deleted. The score of the alignment is: \\
              12. $\texttt{gap\_cost} + D(i-3,j)$ 
            \end{enumerate}
          \item  {\bf If $A[i-2~..~i]$ is not grouped}, then the codon $A[i-2~..~i]$
            is a FS$^+$ codon, and the score of the alignment is:\\
        13. $D(i-1,j) + \texttt{fs\_open\_cost}$\\  
        \end{enumerate}

      \item {\bf If $B[i]$ is aligned with a gap}, there are fives cases that
        are symmetric
        to the five cases from (b):\\
        14. $\frac{s_{an}(A[i],B[j-1])}{2} + \frac{s_{an}(A[i-1],B[j-2])}{2} + D_F(i-2,j-3) + \texttt{fs\_open\_cost}$\\
        15. $s_{an}(A[i],B[j-1]) + D(i-1,j-3) + 2 * \texttt{fs\_open\_cost}$\\  
        16. $\frac{s_{an}(A[i],B[j-2])}{2} + D_F(i-1,j-3) + \texttt{fs\_open\_cost}$\\
         17. $\texttt{gap\_cost} + D(i,j-3)$\\
         18. $D(i,j-1) + \texttt{fs\_open\_cost}$\\  
      \end{enumerate}

    \item {\bf If $i (mod~3) = 0$ and $j (mod~3) \neq 0$}, then $A[i]$ is the
      last nucleotide of a codon $A[i-2~..~i]$ and $B[j]$ is not the last
      nucleotide of a codon. There are three cases depending on the alignment
      of $A[i]$ and $B[j]$.
      \begin{enumerate}
      \item {\bf If  $A[i]$ and $B[j]$ are aligned together},  there are
        two cases depending on whether $A[i-2~..~i]$  is grouped in the
        alignment or not.
        \begin{enumerate}
        \item {\bf If $A[i-2~..~i]$ is grouped}, there are three cases depending
          on the number of nucleotides from $B$ that are aligned with
          $A[i-2~..~i]$, three, two, or one:
          \begin{enumerate}
          \item  {\bf If  $A[i-2~..~i]$ is aligned with three nucleotides}, then
            these nucleotides are $B[j]$, $B[j-1]$, and $B[j-2]$. We are in
            the case of an unmatching (U) codon. The score of the alignment is
            then:\\      
          1. $\frac{s_{aa}(A[i-2~..~i],B[j-2~..~j])}{2} + D_F(i-3,j-3) + \texttt{fs\_extension\_cost}$ + $\frac{s_{an}(A[i],B[j])}{2} ~(+ \frac{s_{an}(A[i-1],B[j-1])}{2} ~if ~j-1 (mod~3) \neq 0)$
          \item {\bf If  $A[i-2~..~i]$ is aligned with two nucleotides}, then
            these nucleotides are $B[j]$ and $B[j-1]$. $A[i-2~..~i]$ is a FS$^-$
            codon. There are two cases
            depending of the alignment of $B[j-1]$ with $A[i-1]$ or $A[i-2]$.
            In both cases, if $j-1 (mod~3) = 0$, then $j-1$ is the last
            nucleotide of a codon. We should then make adjustments in order to
            account for the type of this codon (FS$^+$, or unknown type for now):\\      
          2. $s_{an}(A[i],B[j]) + s_{an}(A[i-1],B[j-1]) + D(i-3,j-2) + \texttt{fs\_open\_cost} ~(+ \texttt{fs\_open\_cost} ~if ~j-1 (mod~3) = 0)$\\
          3. $s_{an}(A[i],B[j]) + s_{an}(A[i-2],B[j-1]) + D_F(i-3,j-2) + \texttt{fs\_open\_cost} ~(- \frac{s_{an}(A[i-2],B[j-1])}{2} ~ if ~ j-1 (mod~3) = 0)$
        \item {\bf If  $A[i-2~..~i]$ is aligned with one nucleotide}, then
          $A[i-2~..~i]$ is a FS$^-$ codon. The score of the alignment is:\\      
          4. $s_{an}(A[i],B[j]) + D(i-3,j-1) + \texttt{fs\_open\_cost}$
          \end{enumerate}
        \item {\bf If $A[i-2~..~i]$ is not grouped}, then $A[i-2~..~i]$ is a
          FS$^+$ codon:\\
         5. $s_{an}(A[i],B[j]) + D(i-1,j-1) + \texttt{fs\_open\_cost}$
        \end{enumerate}
        
      \item {\bf If  $A[i]$ is aligned with a gap},  there are
        two cases depending on whether $A[i-2~..~i]$  is grouped in the
        alignment or not.
        \begin{enumerate}
        \item {\bf If $A[i-2~..~i]$ is grouped}, there are three cases depending
          on the number of nucleotides from $B$ that are aligned with
          $A[i-2~..~i]$, two, one, or zero.
          \begin{enumerate}
          \item {\bf If  $A[i-2~..~i]$ is aligned with two nucleotides}, then
            these nucleotides are $B[j]$ and $B[j-1]$. $A[i-2~..~i]$ is a
            FS$^-$ codon. If $j-1 (mod~3) = 0$, then $j-1$
            is the last nucleotide of a codon. We should make adjustments
            in order to account for the fact no type has yet been decided
            for this codon.\\      
         6. $ s_{an}(A[i-1],B[j]) + s_{an}(A[i-2],B[j-1]) + D_F(i-3,j-2) + \texttt{fs\_open\_cost} ~(- \frac{s_{an}(A[i-2],B[j-1])}{2} ~if ~j-1 (mod~3) = 0)$
          \item {\bf If  $A[i-2~..~i]$ is aligned with one nucleotide}, then
            this single nucleotide is $B[j]$.
            $A[i-2~..~i]$ is a FS$^-$ codon. There are two cases depending
            on the alignment of $B[j]$ with $A[i-1]$ or $A[i-2]$:\\      
         7. $s_{an}(A[i-1],B[j]) + D(i-3,j-1) + \texttt{fs\_open\_cost}$\\
         8. $s_{an}(A[i-2],B[j]) + D(i-3,j-1) + \texttt{fs\_open\_cost}$
       \item {\bf If $A[i-2~..~i]$ is aligned with zero nucleotide}, the codon
        $A[i-2~..~i]$ is entirely deleted:\\      
         9. $\texttt{gap\_cost} + D(i-3,j)$
          \end{enumerate}
        \item {\bf If $A[i-2~..~i]$ is not grouped}\\
         10. $D(i-1,j) + \texttt{fs\_open\_cost}$
       \end{enumerate}
      \item {\bf If  $B[j]$ is aligned with a gap}, then the score of the
        alignment is:\\
         11. $D(i,j-1)$\\
    \end{enumerate}

      \item {\bf If  $i (mod~3) \neq 0$ and $j (mod~3) = 0$}, the proof is
        symmetric to the previous proof for 2.\\
        
      \item {\bf If  $i (mod~3) \neq 0$ and $j (mod~3) \neq 0$}, there are
        three cases depending on the alignment of $A[i]$ and $B[j]$.
      \begin{enumerate}
      \item {\bf If  $A[i]$ and $B[j]$ are aligned together}, the score of the
        alignment is:\\
      1. $s_{an}(A[i],B[j]) + D(i-1,j-1)$
      \item {\bf If  $A[i]$ is aligned with a gap}, the score of the
        alignment is:\\
      2. $D(i-1,j)$
      \item {\bf If  $B[j]$ is aligned with a gap}, the score of the
        alignment is:\\
      3. $D(i,j-1)$    
      \end{enumerate}

  \end{enumerate}
\end{proof}

\subsection*{Proof of Lemma \ref{DF}}

\begin{proof}[Proof of Lemma \ref{DF}] The proof follows from Lemma \ref{D}.
  \begin{enumerate}
  \item {\bf If $i (mod~3) = 0$ and $j (mod~3) = 0$}, this case is trivial.
  \item {\bf If $i (mod~3) = 2$ and $j (mod~3) = 0$}, then $i+1 (mod~3) = 0$
    and $j+1 (mod~3) = 1 \neq 0$. The five cases follow from the
    application of Lemma \ref{D}, Case 2 for computing $D(i+1,j+1)$,
    and by keeping
    only the cases where $A[i+1]$ and $B[i+1]$ are aligned together (cases 1,
    2, 3, 4, 5 among the 11 cases). However, in each of the cases,
    we must subtract half of the score of aligning $B[i+1]$ with $A[i+1]$
    ($\frac{s_{an}(A[i+1],B[j+1])}{2}$), because this score will be added
    subsequently.
    
  \item {\bf If $i (mod~3) = 0$ and $j (mod~3) = 2$}, the proof is symmetric to the previous case.

  \item {\bf If $i (mod~3) = 1$ and $j (mod~3) = 0$}, then $i+2 (mod~3) = 0$
    and $j+2 (mod~3) = 2 \neq 0$. Here again, the three cases follow
    from the application of Lemma \ref{D}, Case 2 for computing
    $D(i+2,j+2)$, and
    by keeping only the cases where $A[i+1]$, $B[i+1]$, and $A[i+1=2]$,
    $B[i+2]$ can be aligned together (cases 1, 2, 5 among the 11 cases).
    However, in each of the cases, we must subtract half of the
    scores of aligning $B[i+2]$ with $A[i+2]$, and aligning $B[i+1]$ with
    $A[i+1]$ ($\frac{s_{an}(A[i+2],B[j+2])}{2}$,
    $\frac{s_{an}(A[i+1],B[j+1])}{2}$), because theses scores will be added
    subsequently.
    
  \item {\bf If $i (mod~3) = 0$ and $j (mod~3) = 1$}, the proof is symmetric to the previous case.
  \end{enumerate}
  \end{proof}

\subsection*{Proof of Theorem \ref{thm}}
\begin{proof}[Proof of Theorem \ref{thm}]
  The proof relies on two points: (1) The algorithm computes the maximum
  score of an alignment between $A$ and $B$, and (2) the algorithm runs
  with an $O(n.m)$ time and space complexity.\\

  \noindent
  (1) The validity of the algorithm, i.e. the facts that it fills the cells of
  the tables $D$, $D_F$ according to Definition \ref{table}, follows from
  five points.
  \begin{itemize}
  \item The initialization of the tables is a direct consequence of
  Definition \ref{table}.
  \item Lemmas \ref{D} and \ref{DF}.
  \item The couples $(i,j)$ of prefixes of $A$ and $B$ that need to be
    considered
    in the algorithm are all the possible couples for $D(i,j)$, and only
    the couples such
  that $i (mod~3) = 0$ or $j (mod~3) = 0$ for $D_F(i,j)$ (see all the cases
  in which the table $D_F$ is used in Lemmas \ref{D} (7 cases) and \ref{DF}
  (3 cases)).
  \item The couples $(i,j)$ of prefixes of $A$ and $B$ are considered
  in increasing
  order of length, and $D[i,j]$ is computed before $D_F[i,j]$ in the cases
  where $i (mod~3) = 0$ or $j (mod~3) = 0$.
  \item A backtracking of the algorithm allows to find a maximum
  score alignment between $A$ and $B$.
  \end{itemize}

  \noindent
  (2) The time and space complexity of the algorithm is a direct consequence
  of the number of cells of the tables $D$ and $D_F$,
  $2 \times (n+1 \times m+1)$. Each cell is filled in constant time.
  \end{proof}

\subsection*{Alignment of coding sequences NM\_001083537  and NM\_018172 using
a previously published method \cite{ranwez2011} and the present method}

\begin{verbatim}
Translations of NM_001083537 and NM_018172 into protein sequences

>NM_001083537
MAPEENAGTELLLQGFERRFLAVRTLRSFPWQSLEAKLRDSSDSELLRDILQKTVRHPVC
VKHPPSVKYAWCFLSELIKKSSGGSVTLSKSTAIISHGTTGLVTWDAALYLAEWAIENPA
AFINRTVLELGSGAGLTGLAICKMCRPRAYIFSDPHSRVLEQLRGNVLLNGLSLEADITG
NLDSPRVTVAQLDWDVAMVHQLSAFQPDVVIAADVLYCPEAIVSLVGVLQRLAACREHKR
APEVYVAFTVRNPETCQLFTTELGRDGIRWEAEAHHDQKLFPYGEHLEMAMLNLTL*

>NM_018172
MAPEENAGSELLLQSFKRRFLAARALRSFRWQSLEAKLRDSSDSELLRDILQKHEAVHTE
PLDELYEVLVETLMAKESTQGHRSYLLTCCIAQKPSCRWSGSCGGWLPAGSTSGLLNSTW
PLPSATQRCASCSPPSYAGLGSDGKRKLIMTRNCFPTESTWRWQS*
\end{verbatim}

\begin{verbatim}

Score obtained with previously published method : 718.0

A!!TGGCGCCCGAGGAGAACGCGGGGACCGAACTCTTGCTGCAGGGTTTTGAGCGCCGCT
A!!TGGCGCCCGAGGAGAACGCGGGGAGCGAACTCTTGCTGCAGAGTTTCAAGCGCCGCT

TCCTGG---CGGTGCGCACACTGCGCTCCTTC!CCC---TGGCAGAGCTTAGAGGCAAAG
TCCTGGCAGCGC---GCGCCCTGCGCTCCTT!!CCGC!!TGGCAGAGCTTAGAAGCAAAG

TTAAGAGACT!!CATCAGATTCTGAGCTGCTGCGGGATATTTTGCAGAAGACTGTGAGGC
TTAAGAGACT!!CATCAGATTCTGAGCTGCTGCGGGATATTTTGC---AGA---------
\end{verbatim}
\newpage
\begin{verbatim}
ATCCTGTGTGTGTGAAGCACCCGCCG!TCAGTCAAGTATGCCTGGTGCT!!TTCTCTCAG
---------------AGCACGAGGC!!------------------TGT------------

AACTCATCAAAAAGTCCTCAGGAGGCTCAGTCACACTCTCCAAGAGCACAGCCATCATCT
---------------CCA---------------CAC------AGAG!!---CCTT!!---

CCCACGGTACCACAGGCCTGGTCACATGGGATGCCGCCCTCTA!!CCTTGCAGAATGGGC
------------------TGG---------ATG------AGC------TGT---ACG---

CATCGAGAACCCGGCAGCCTTCATTAACAGGACTGTCCTAGAGCTTGGCAGTGGTGCCGG
---AGG---------------------------TGC---------TGG---TGG---AGA

CCTCACAGGCCTTGCCATCTGCAAGATGTGCCGCCCCCGGGCATACATCTTCAGCGACCC
C!!------CCT------------GAT------------GGC------------------

TCACAGCCGGGTCCTCGAGCAGCTCCGAGGGAATGTCCTTCTCAATGGCCTCTCATTAGA
---CAA---GGA------------------------------------------GTC---

GGCAGACATCACTGGCAACTTAGACAGCCCCAGGGTGACAGTGGCCCAGCTGGACTGGGA
---------CAC---------------CCA------------GGGCCA------------

CGTAGCAATGGTCCAT!!CAGCTCTCTGCCTTCCAGCCAGATGTTGTCATTGCAGCAGAC
------------CCG---GAGCTA------TTT---------------------GCTGAC

G!!TGCTGTATTGCCCAGAAGCCATCGTGTCGCTGGTCGGGGTCCTGCAGAGGCTGGCTG
G!!TGCTGTATTGCCCAGAAGCCATCGTGTCGCTGGTCGGGGTCCTGCGGAGGCTGGCTG

CCTGCCGGGAGCACAAGCGGGCTCCTGAGGTCTACGTGGCCTTTACCGTCCG!!CAACCC
CCTGCCGGGAGCACCAGCGGGCTCCTCAATTCTACATGGCCCTTACCGTCTG!!CAACCC

AGAGACGTGCCAGCTGTTCACCACCGAGCTAG!GCCGGGA!!TGGGATC!!AGATGGGAA
AGAGATGTGCCAGCTGTTCACCACCGAGCTAT!GCTGGAC!!TGGGATC!!AGATGGGAA

GCGGAAGCTCATCATGACCAGAAACTGTTTCCCTATG!!GAGAGCACTTGGAGATGGCAA
GCGGAAGCTCATCATGACCAGAAACTGTTTCCCTACA!!GAGAGCACTTGGAGATGGCAA

TGCTGAACCTCACACTGTAG!
---------AGC------TGA
\end{verbatim}

\begin{verbatim}
Score obtained with present method: 530.0

ATGGCGCCCGAGGAGAACGCGGGGACCGAACTCTTGCTGCAGGGTTTTGAGCGCCGCTTC
ATGGCGCCCGAGGAGAACGCGGGGAGCGAACTCTTGCTGCAGAGTTTCAAGCGCCGCTTC

CTGGCGGTGCGCACACTGCGCTCCTTCCCCTGGCAGAGCTTAGAGGCAAAGTTAAGAGAC
CTGGCAGCGCGCGCCCTGCGCTCCTTCCGCTGGCAGAGCTTAGAAGCAAAGTTAAGAGAC

TCATCAGATTCTGAGCTGCTGCGGGATATTTTGCAGAAGACTGTGAGGCATCCTGTGTGT
TCATCAGATTCTGAGCTGCTGCGGGATATTTTGCAG------------------------

GTGAAGCACCCGCCGTCAGTCAAGTATGCCTGGTGCTTTCTCTCAGAACTCATCAAAAAG
---AAGCAC------------------------------------GAG------------

TCCTCAGGAGGCTCAGTCACACTCTCCAAGAGCACAGCCATCATCTCCCACGGTACCACA
------------------------------------GCT---GTC---CAC---ACA-GA

GGCCTGGTCACATGGGATGCCGCCCTCTACCTTGCAGAATGGGCCATCGAGAACCCGGCA
G-CCT-TTG------GAT---GAGCTGTAC------GAG------GTG------------

GCCTTCATTAACAGGACTGTCCTAGAGCTTGGCAGTGGTGCCGGCCTCACAGGCCTTGCC
------CTG---------GTG---GAG---------------------ACC---CTG---

ATCTGCAAGATGTGCCGCCCCCGGGCATACATCTTCAGCGACCCTCACAGCCGGGTCCTC
---------ATG------------GCC---------------------------------

GAGCAGCTCCGAGGGAATGTCCTTCTCAATGGCCTCTCATTAGAGGCAGACATCACTGGC
---------AAG------------------------------GAG---------------

AACTTAGACAGCCCCAGGGTGACAGTGGCCCAGCTGGACTGGGACGTAGCAATGGTCCAT
---------TCC---------ACC------CAG--------------------GGC-CAC

CAGCTCTCTGCCTTCCAGCCAGATGTTGTCATTGCAGCAGACGTGCTGTATTGCCCAGAA
CGG---AGC---TAT----------------TTGCT---GACGTGCTGTATTGCCCAGAA

GCCATCGTGTCGCTGGTCGGGGTCCTGCAGAGGCTGGCTGCCTGCCGGGAGCACAAGCGG
GCCATCGTGTCGCTGGTCGGGGTCCTGCGGAGGCTGGCTGCCTGCCGGGAGCACCAGCGG

GCTCCTGAGGTCTACGTGGCCTTTACCGTCCGCAACCCAGAGACGTGC--CAGCTGTTCA
GCTCCTCAATTCTACATGGCCCTTACCGTCTGCAACCCAGAGA--TGTGCCAGCTGTTCA

CCACCGAGCTA----GGCCGGGATGGGATCAGATGGGAAGCGGAAGCTCATCATGACCAG
CCACCGAGCTATGCTGGA----CTGGGATCAGATGGGAAGCGGAAGCTCATCATGACCAG

AAACTGTTTCCCTATGGAGAGCACTTGGAGATGGCAATGCTGAACCTCACACTGTAG
AAACTGTTTCCCTACAGAGAGCACTTGGAGATGGCAA-----------AGC---TGA
\end{verbatim}

\end{document}